\documentclass[conference]{IEEEtran}
\IEEEoverridecommandlockouts

\usepackage{ifpdf}

\ifCLASSINFOpdf
  \usepackage[pdftex]{graphicx}
\else
\fi
\usepackage{amsmath}
\usepackage{amsfonts}
\usepackage{amssymb}
\usepackage{amsthm}
\usepackage{algorithmic}

\usepackage{array}
\usepackage{float}
\usepackage{url}
\usepackage{xcolor}
\newtheorem{theorem}{Theorem}
\newtheorem{definition}{Definition}
\newtheorem{example}{Example}
\newtheorem{proposition}{Proposition}

\DeclareMathOperator*{\cbar}{\overline{\mathcal{C}}}

\begin{document}
%
\title{Tiling of Constellations}


\author{\IEEEauthorblockN{Maiara F. Bollauf and \O yvind Ytrehus}
\vspace{0.2cm}
\IEEEauthorblockA{Simula UiB,
N-5008, Bergen, Norway \\
Email: {\{maiara, oyvindy\}}@simula.no}}



\maketitle

\begin{abstract} Motivated by applications in reliable and secure communication, we address the problem of tiling (or partitioning) a finite constellation in $\mathbb{Z}_{2^L}^n$ by subsets, in the case that the constellation does not possess an abelian group structure. The property that we do require is that the  constellation is generated by a linear code through an injective mapping. The intrinsic relation between the code and the constellation provides a sufficient condition for a tiling to exist. We also present a necessary condition. Inspired by a result in group theory, we discuss results on tiling for the particular case when the finer constellation is an abelian group as well. 
\end{abstract}




%
\IEEEpeerreviewmaketitle

\vspace{0.1cm}

\section{Introduction}

	The problem of tiling consists of factorizing (or partitioning) a group into sets such that their intersection is just the identity element of the group. This problem attracted a lot of attention since Haj\'os \cite{hajos49} reduced the Minkowski study of filling $\mathbb{R}^n$ by congruent cubes to the factorization of finite groups. Since then, several authors \cite{newman77, redei65, sands79, szabo09} contributed towards this solution, but many circumstances are still unsolved.
	
	Besides the mathematical interest around tiling, including the estimation of Ramsey number of graphs \cite{szabo09}, applications in coding theory are also well established. Connections between tilings and perfect codes were addressed in \cite{cohen96}, for tilings of $\mathbb{F}_2^n$ with the Hamming metric and in \cite{strey21}, for tilings of general lattices by their sublattices with the Euclidean metric. Tilings of finite abelian groups were also explored in \cite{qu94} in the context of public key cryptography.

	The motivation to our study comes from the wiretap channel, which is a model of communication that concerns both reliability and security. First introduced by Wyner \cite{wyner75}, the assumption is that the transmitter aims to send a message to a receiver through a main channel, which is connected to an eavesdropper. The general idea is to construct an encoder and decoder to jointly optimize the data transmission and the equivocation seen by the illegitimate receiver. This tradeoff, denoted by secrecy capacity, was proved to be achieved with random coding arguments and an encoding scheme called coset encoding, where the transmitter encodes the message in terms of the corresponding \emph{syndrome}. Several papers addressed this problem considering a variety of communication channels and coding schemes, including \cite{mahdavifar11, ozarov84, thangaraj07}.
	
	Tiling a vector space is always possible. For that reason, linear codes make for a natural alternative for coset encoding through their cosets. On the other hand, lattices with their group structure also admit such characterization and they have been used to demonstrate security properties on the  wiretap Gaussian channel, as discussed in \cite{ling14, oggier15}. Nevertheless, the receiver is supposed to perform a lattice decoding, which is known to be a hard problem. One approach to the lattice decoding problem is to select a constellation that allows an efficient decoding on the the main channel, for example, Construction C (or multilevel code formula \cite{forney88}), denoted by $\Gamma_C.$ Asymptotically, Construction C achieves  capacity at high SNR on the AWGN channel  with multistage decoding \cite{forney00}.
	
	\emph{Contributions:} In general, $\Gamma_C$ produces a nonlattice constellation, generated by a finite set in $\mathbb{Z}_{2^L}^n$. Recently a subset $\Gamma_{C^\star} \subseteq \Gamma_C$ was introduced \cite{bzc2019}, by associating the linear codes underlying both constructions. In this work, we discuss conditions for the existence of a tiling of $\Gamma_C$ by classes of $\Gamma_{C^\star}.$ 
	One condition emerges from the linear code generating such constellations, as we can always tile a code with its cosets, but the second approach involves solving a pure mathematical tiling problem and taking into account the particular case where $\Gamma_C$ is an abelian group (or a lattice). While motivated by the potential application to the wiretap channel, the theoretical problem itself is the emphasis of this paper.  

	This paper is organized as follows: Sec.~\ref{Sec:2} presents preliminary definitions of codes, lattices, and also a way to map codes to constellations such as $\Gamma_C.$ Sec.~\ref{Sec:3} recalls tiling properties of linear codes. Sec.~\ref{Sec:4} explores a condition inherited from the linear code that guarantees the tiling for the most general form of $\Gamma_C$ (lattice and nonlattice). Sec.~\ref{Sec:5} focuses on the case where $\Gamma_C$ is an abelian group (lattice). Finally,  Sec.~\ref{Sec:Conc} concludes the paper.



 

\section{Mapping linear codes to a constellation} \label{Sec:2}

	There are two essential mathematical concepts that underlie our work, which are linear codes and lattices.

\begin{definition} A set $\mathcal{C} \subseteq \mathbb{F}_2^n$ is a binary linear $[n,k]-$code if it is a $k-$dimensional subspace of $\mathbb{F}_2^n.$ 
\end{definition}

\begin{definition} A lattice $\Lambda \subset \mathbb{R}^n$ is a discrete abelian subgroup of $\mathbb{R}^n.$
\end{definition}

    	From now on, we will admit a linear code $\mathcal{C} \subset \mathbb{F}_2^{nL}$ and split it into $L$  \emph{projection} codes $\mathcal{C}_1,\ldots,\mathcal{C}_L,$ where each $\mathcal{C}_i \subseteq \mathbb{F}_2^{n}$ is the restriction of $\mathcal{C}$ to coordinate positions $(i-1)n +1,\ldots,in,$ where $i=1, \dots, L.$ Accordingly, each codeword ${\bf c} \in \mathcal{C}$ can be written as ${\bf c} = ({\bf c}_1, \dots, {\bf c}_L)$ with ${\bf c}_i \in \mathcal{C}_i $.

	The mapping $\psi: \mathbb{F}_2^{nL} \rightarrow \mathbb{Z}_{2^L}^n$ defined as 
\begin{align}\label{eq_psi}
\psi({\bf c}_1, \dots, {\bf c}_L)  = ~ {\bf c}_1+2{\bf c}_{2}+\dots+2^{L-1}{\bf c}_{L},
\end{align}
associates a code $\mathcal{C} \subset \mathbb{F}_2^{nL}$ to a finite constellation (not necessarily an abelian group), where each ${\bf c}_i \in \mathbb{F}_2^n.$ 

\begin{definition}\cite{bzc2019}\label{def_cstar} Consider a linear code $\mathcal{C} \subset \mathbb{F}_2^{nL}.$ The infinite constellation defined by $\Gamma_{C^\star} = \psi(\mathcal{C})+2^L\mathbb{Z}^n = Y + 2^L\mathbb{Z}^n$ is called Construction $C^\star.$ 
\end{definition}
	
	
\begin{definition}\cite{forney88} \label{def_c} Let $\overline{\mathcal{C}} = \mathcal{C}_1 \times \dots \times \mathcal{C}_L,$ where each $\mathcal{C}_i \subseteq \mathbb{F}_2^{n}$ is the respective projection code of a linear code $\mathcal{C} \subset \mathbb{F}_2^{nL}$, for $i=1,\dots,L.$ We denote by Construction C the set $\Gamma_C = \psi(\overline{\mathcal{C}})+2^L\mathbb{Z}^n = X + 2^L\mathbb{Z}^n.$ 
\end{definition}	 

   Throughout the paper, the subsets $X$ and $Y$ would always refer to $X=\psi(\mathcal{\overline{C}})$ and $Y=\psi(\mathcal{C}),$ for $\overline{\mathcal{C}},\mathcal{C} \subseteq \mathbb{F}_2^{nL}$ as in Definitions \ref{def_cstar} and \ref{def_c}.
    
\begin{example} 
\label{ex:simple}
Consider $\mathcal{C} \subset \mathbb{F}_2^{6},$ with $n=2, L=3,$ given by
\begin{align}
\mathcal{C} = & \{ (0,0,0,0,0,0),(0,0,0,1,0,0),(1,1,0,0,0,0), \nonumber \\
& ~ (1,1,0,1,0,0),(0,0,1,0,1,0),(0,0,1,1,1,0), \nonumber \\
& ~ (1,1,1,0,1,0),(1,1,1,1,1,0)\}.
\end{align}
Here,  generator matrices of the codes  $\mathcal{C}$,  $\overline{\mathcal{C}}$, $\mathcal{C}_1$, $\mathcal{C}_2$, and $\mathcal{C}_3$ are given (respectively) by
\begin{align*}
G = & \begin{pmatrix}
1 & 1 & 0 & 0 & 0 & 0 \\
0 & 0 & 0 & 1 & 0 & 0 \\ 
0 & 0 & 1 & 0 & 1 & 0 
\end{pmatrix}, ~
\overline{G} = \begin{pmatrix}
1 & 1 & 0 & 0 & 0 & 0 \\
0 & 0 & 0 & 1 & 0 & 0 \\ 
0 & 0 & 1 & 0 & 0 & 0  \\
0 & 0 & 0 & 0 & 1 & 0
\end{pmatrix},
\end{align*}
\vspace{-0.3cm}
\begin{align*}
G_1  = & \begin{pmatrix}
1 & 1 
\end{pmatrix}, ~
G_2 =  \begin{pmatrix}
 0 & 1  \\ 
 1 & 0  
\end{pmatrix}, ~
G_3 = \begin{pmatrix}
1 & 0 
\end{pmatrix}.
\end{align*}

    Here $\Gamma_{C^\star}=\psi(\mathcal{C}) +  8\mathbb{Z}^2 = Y + 8\mathbb{Z}^2,$ where $Y=\{(0,0),$ $(0,2),(1,1),(1,3),(6,0),(6,2),(7,1),$ $(7,3)\}.$ And $\Gamma_C=\psi(\cbar)+8\mathbb{Z}^2 = X + 8\mathbb{Z}^2,$ where $X=Y \cup \{(4,0),(4,2),(2,0),$ $(2,2),(5,1),(5,3),(3,1),(3,3)\}.$
\end{example}

	Observe that if $\mathcal{C}=\overline{\mathcal{C}},$ then $\Gamma_{C^\star}= \Gamma_C,$ however for any other case $\Gamma_{C^\star}\subset \Gamma_C.$ Under general assumptions, both $\Gamma_{C^\star}$ and $\Gamma_{C}$ are nonlattice constellations. To discuss necessary and sufficient conditions such that $\Gamma_C$ is a lattice, we need to define the Schur (or coordinate-wise) product.
	
\begin{definition} The Schur product of ${\bf x}=(x_1, \dots, x_n),~ {\bf y}=(y_1, \dots, y_n) \in \mathbb{F}_2^{n}$ is ${\bf x} \ast {\bf y} = (x_1y_1, \dots, x_ny_n) \in \mathbb{F}_2^n.$ 
\end{definition}

    A chain of linear codes $\mathcal{C}_1 \subseteq \dots \subseteq \mathcal{C}_L$ is said to be closed under Schur product if for any two elements ${\bf c}_i, {\bf \tilde{c}}_i \in \mathcal{C}_i,$ it is valid that ${\bf c}_i \ast {\bf \tilde{c}}_i \in \mathcal{C}_{i+1},$ for all $i=1, \dots, L-1.$ In Definition~\ref{def_c}, $\Gamma_C$ is a lattice if and only $\mathcal{C}_1 \subseteq \dots \subseteq \mathcal{C}_L$ and this chain is closed under Schur product \cite{kosit14}. 
    
    Next, we are going to express operations in $\mathbb{Z}_{2^L}^n$ in terms of operations in $\mathbb{F}_2^n.$ 
	
	If $\phi: \mathbb{F}_2^n \rightarrow \mathbb{Z}^n$ is the natural embedding, then 
	$\phi({\bf x})+\phi({\bf y}) = \phi({\bf x} \oplus {\bf y}) + 2 \phi({\bf x} \ast {\bf y}),$ where $\oplus$ denotes the addition in $\mathbb{F}_2^n$ and $+$ denotes the addition in $\mathbb{Z}^n.$ We will write ${\bf x} + {\bf y} = {\bf x} \oplus {\bf y} + 2 ({\bf x} \ast {\bf y})$ to simplify the notation.

\begin{proposition}\label{prop1} Consider $\mathcal{C} \subseteq \mathbb{F}_2^{nL}$ a linear code and the mapping $\psi$ defined as in~\eqref{eq_psi}. Then, for ${\bf c},{\bf \tilde{c}} \in \mathcal{C},$
\[
\psi({\bf c}) + \psi({\bf \tilde{c}})= \psi({\bf c} \oplus {\bf \tilde{c}})  + 2 \psi({\bf c} \ast {\bf \tilde{c}}),
\]
\end{proposition}

\begin{proof} We have by definition that
\begin{align*}
\psi({\bf c} \oplus {\bf \tilde{c}}) = & ({\bf c}_1 \oplus {\bf \tilde{c}}_1) + \dots + 2^{L-1}({\bf c}_{L-1} \oplus {\bf \tilde{c}}_{L-1}) \\
=  & ({\bf c}_1 + {\bf \tilde{c}}_1 - 2({\bf c}_1 \ast {\bf \tilde{c}}_1)) + \dots + \\
& + 2^{L-1}({\bf c}_{L-1} + {\bf \tilde{c}}_{L-1} -2({\bf c}_{L-1} \ast {\bf \tilde{c}}_{L-1} )) \\
= & ({\bf c}_1 + \dots + 2^{L-1}{\bf c}_{L-1}) + ({\bf \tilde{c}}_1 + \dots + 2^{L-1}{\bf \tilde{c}}_{L-1})\\
& - 2(({\bf c}_1 \ast {\bf \tilde{c}}_1) + \dots + 2^{L-1}({\bf c}_{L-1} \ast {\bf \tilde{c}}_{L-1})) \\
= & \psi({\bf c}) + \psi({\bf \tilde{c}}) -2 \psi({\bf c} \ast {\bf \tilde{c}}).  \tag*{\qedhere}
\end{align*}
\end{proof}

	Besides this operation, it will be helpful to establish the notion of level shift.

\begin{definition} Given a codeword ${\bf c}=({\bf c}_1, \dots, {\bf c}_L) \in \mathcal{C} \subset \mathbb{F}_2^{nL},$ we define $${\bf c}^{(i)} = (\underbrace{{\bf 0}, \dots, {\bf 0}}_{i},{\bf c}_1, \dots, {\bf c}_{L-i-1},  {\bf c}_{L-i})$$ as the ($i-$th) level shift of the codeword ${\bf c},$ for any $i=1,\dots, L-1.$
\end{definition}

    For example, ${\bf c}^{(1)}=({\bf 0},{\bf c}_1, \dots, {\bf c}_{L-1})$ and ${\bf c}^{(L-1)}=({\bf 0},{\bf 0}, \dots,{\bf c}_1).$ From this definition, it is clearly valid that $2 \psi({\bf c}) = \psi({\bf c}^{(1)}).$ Throughout the paper, we will also refer to the \textit{Schur level shift}, which means that, given ${\bf c} \ast {\bf \tilde{c}} = ({\bf c}_1 \ast {\bf \tilde{c}}_1, \dots, {\bf c}_L \ast {\bf \tilde{c}}_L),$ then $({\bf c} \ast {\bf \tilde{c}})^{(1)} = ({\bf 0},{\bf c}_1 \ast {\bf \tilde{c}}_1 , \dots, {\bf c}_{L-1} \ast {\bf \tilde{c}}_{L-1}),$ for ${\bf c}, {\bf \tilde{c}} \in \mathbb{F}_2^{nL}.$
	
	From Proposition~\ref{prop1}, it might also be necessary to work with the representation of a negative number in $\mathbb{Z}_{2^L}^n.$ For this reason, we use the \textit{two's complement} as follows
\begin{align} \label{eq_negative0}
-\psi({\bf c}) = & ~ \psi((1,\dots,1) \oplus {\bf c}) + \psi(1,0,\dots,0) \nonumber \\
= & ~ \psi((1,\dots,1) \oplus {\bf c} \oplus (1,0,\dots,0)) + \nonumber \\
& ~ 2\psi(\underbrace{((1,\dots,1) \oplus {\bf c}) \ast (1,0,\dots,0)}_{{\bf u}}).
\end{align}
Then, ${\bf u}={\bf 0}$ or ${\bf u}=(1,0,\dots,0)$ and
\begin{align}\label{eq_negative}
-\psi({\bf c}) = & \begin{cases} 
\psi((0,1,\dots, 1) \oplus {\bf c}), ~\text{if}~ {\bf u} = {\bf 0}  \\
\psi((0,1,\dots, 1) \oplus {\bf c} ) + \psi(\underbrace{0,1,\dots,0}_{(1,0,\dots,0)^{(1)}}), ~\text{otherwise}.
\end{cases}
\end{align}	

	The overall idea is to write $-\psi({\bf c})=\psi({\bf c'}),$ for a given ${\bf c} \in \mathcal{C}$ and ${\bf c'} \in \cbar.$ If the first case of~\eqref{eq_negative} holds, then immediately ${\bf c'} = (0,1,\dots, 1) \oplus {\bf c}.$ Otherwise, one needs to apply recursively the result of Proposition~\ref{prop1} to $\psi((0,1,\dots, 1) \oplus {\bf c}) + \psi(0,1,\dots,0),$ until the Schur product vanishes and the value of ${\bf c'}$ is found.

\section{Tiling of $\cbar$ with cosets of $\mathcal{C}$} \label{Sec:3}

	We start by defining formally tiling of finite abelian groups.
	
\begin{definition}\label{def_tiling}\cite{dinitz06} A tiling of a finite abelian group $G \subset \mathbb{R}^n$ is a pair $(A,B)$ of subsets of $G$ such that $A \cap B = \{\bf 0\}$ and every ${\bf g} \in G$ can be uniquely written as ${\bf g}={\bf a}+{\bf b},$ ${\bf a} \in A,$ ${\bf b} \in B.$ We then say that $(A,B)$ is a tiling of $G.$ 
\end{definition} 
	
	Let $|\mathcal{C}|=2^m$ and $|\cbar| = |\mathcal{C}_1 \times \dots \times \mathcal{C}_L|=2^{\overline{m}},$ where $m<\overline{m} = \sum_{i=1}^L \dim (\mathcal{C}_i)$.

	Due to the fact that $\mathcal{C}$ is an abelian subgroup of the additive group $\cbar,$ we can conclude that cosets of $\mathcal{C}$ tile $\cbar$ into $2^{\overline{m}-m}$ disjoint sets of cardinality $2^m$ each. It allows us to write
    $\cbar = \mathcal{C} \oplus \mathcal{D},$ where $\mathcal{C} \cap \mathcal{D} = \{{\bf 0} \},$ $\mathcal{D} \subset (\cbar \setminus \mathcal{C}) \cup \{{\bf 0}\},$ and the elements ${\bf d}_i \in \mathcal{D}$ are called coset representatives, for $i=1,\dots, 2^{\overline{m}-m}.$ Each coset will have only one representative ${\bf d}_i \in \mathcal{D}.$ 
	
	Observe that any choice of coset representative can be considered, however, for practical reasons, it is usually taken as the minimum weight codeword among the coset and denoted by coset leader. The next theorem summarizes these properties and it is a well known result in linear algebra.

\begin{theorem}\label{thm_coset} Let $\mathcal{C} \subset \cbar = \mathcal{C}_1 \times \dots \times \mathcal{C}_L \subseteq \mathbb{F}_2^{nL}$ be linear codes. Then, the following properties hold:
\begin{enumerate}
\item[i)] Every coset of $\mathcal{C}$ contains exactly $2^{m}$ codewords.
\item[ii)] If ${\bf x} \in (\mathcal{C} \oplus {\bf y}),$ then $(\mathcal{C} \oplus {\bf x}) = (\mathcal{C} \oplus {\bf y}).$ 
\item [iii)] Every codeword in $\cbar$ is contained in one and only one coset of $\mathcal{C}.$
\item[iv)] There are exactly $2^{\overline{m}-m}$ cosets of $\mathcal{C}.$
\end{enumerate}
\end{theorem}

\begin{proof} For item iii), observe that ${\bf x} \in (\mathcal{C} \oplus {\bf x}),$ since ${\bf x} = ({\bf 0} \oplus {\bf x}),$ ${\bf 0} \in \mathcal{C},$ and every codeword of $\cbar$ is contained in some coset of $\mathcal{C}.$ Moreover, suppose ${\bf x} \in (\mathcal{C} \oplus {\bf y}) ~ \cap ~ (\mathcal{C} \oplus {\bf z}),$ then ${\bf x} = {\bf c} \oplus {\bf y} = {\bf c'} \oplus {\bf z}$ and ${\bf y} = 
({\bf c'} \ominus {\bf c}) \oplus {\bf z} \in (\mathcal{C} \oplus {\bf z})$ and by item ii) $(\mathcal{C} \oplus {\bf y}) = (\mathcal{C} \oplus {\bf z}).$  
\end{proof}

\section{Tiling of irregular constellations} \label{Sec:4}

	Unlike codes and lattices, which have the regular structure of vector space and group, respectively, we will now consider constellations that do not possess such a structure, but which are \emph{derived} from linear codes. The idea is to analyze characteristics of such constellations inherited from the linear code.
	
	In this section we will work with the finite representation of the constellations $\Gamma_C$ and $\Gamma_{C^\star},$ given by $X$ and $Y$ respectively (see Definitions~\ref{def_cstar} and \ref{def_c}). Observe that $|X| = |\cbar| = 2^{\overline{m}}$ and $|Y| = |\mathcal{C}| = 2^m.$ By defining a \emph{class} of $\Gamma_{C^\star}$ as the image under $\psi$ of a coset representative of the code $\mathcal{C} \subseteq \mathbb{F}_2^{nL},$ from Theorem~\ref{thm_coset}, we can conclude that there are $|X|/|Y| = 2^{\overline{m}-m}$ distinct classes of $\Gamma_{C^\star}.$
	
	For convenience, we will use the notation of tiling defined for groups as in Definition~\ref{def_tiling} also for the sets $X$ and $Y,$ which in general do not have group structure, but share the same operation (addition). We must then demonstrate that the representation of $X=Y+Z$ as a sum is unique, which is equivalent to show that $(Y+{\bf z}_j) \cap (Y+{\bf z}_k) = \varnothing,$ for ${\bf z}_j \neq {\bf z}_k \in Z,$ and that $Y \cap Z=\{{\bf 0}\}.$ Here also $|Z| = 2^{\overline{m}-m}$ and the elements of $Z$ will be called class representatives. Instead of \emph{tiling}, it might be more appropriate to use the term \emph{partition of sets}. However, we aimed for a unique notation throughout the paper. We apologize for any confusion arising from this abuse of notation.

    In general, $(Y,Z)$ is not a tiling of $X.$
    
\begin{example} Consider $\mathcal{C}=\{(0,0,0,0),(1,1,0,0),(1,0,1,0),$ $(1,1,1,0)\} \subset \mathbb{F}_2^{4},$ $n=1,~L=4.$ Therefore, $Y=\{0,3,5,6\}$ and $X=\{0,1,2,3,4,5,6,7\} \subset \mathbb{Z}_{16}.$ 
There are four possible choices for the nonzero remaining element of $Z,$ which are $1,2,4,7,$ but observe that none of them generate the whole set $X.$ This is explained by the distance profile of $Y,$ which is not constant and cannot be obtained by a unique translation.
\end{example}

Based on the tiling of the underlying codes we can derive the following condition for the tiling of the finite constellation. From the previous section, we have that $(\mathcal{C},\mathcal{D})$ is a tiling of $\cbar,$ ${\bf d}_i \in \mathcal{D}$ are the coset representatives, for $i=1, \dots, 2^{\overline{m}-m}.$ Let $Z = \{{\bf z}_i:  {\bf z}_i= \psi({\bf d}_i),$ $~\text{for the respective coset representatives}~ {\bf d}_i \in \mathcal{D}\}.$ 

\begin{theorem}\label{thm_zeroschur} $(Y,Z)$ is a tiling of $X$ if $({\bf c} \ast {\bf d}_i)^{(1)} = {\bf 0},$ for all ${\bf c} \in \mathcal{C},$ ${\bf d}_i \in \mathcal{D}.$
\end{theorem}	

\begin{proof} 
	Consider ${\bf x} \in X$. By definition, ${\bf x} = \psi({\bf b})$ for some ${\bf b} \in \cbar$. Since $(\mathcal{C}, \mathcal{D})$ is a tiling of $\cbar,$ we can write ${\bf b} = {\bf c} \oplus {\bf d}_i,$ for ${\bf c} \in \mathcal{C},$ ${\bf d}_i \in \mathcal{D},$ $i=1,\dots, 2^{\overline{m}-m}.$ Hence,
\begin{align}
{\bf x} = & ~ \psi({\bf c} \oplus {\bf d}_i) = \psi({\bf c}) + \psi({\bf d}_i) - 2\psi({\bf c} \ast {\bf d}_i) \nonumber \\
 = & ~ \underbrace{\psi({\bf c})}_{\in Y} + \underbrace{\psi({\bf d}_i)}_{\in Z} - \psi\big(\underbrace{({\bf c} \ast {\bf d}_i)^{(1)}\big)}_{{\bf 0},~ \text{by hypothesis}},
\end{align} 
for all ${\bf c} \in \mathcal{C}.$ Clearly $Y \cap Z=\{\bf 0\}.$
Moreover, the intersection of sets indicated by different coset representatives is empty. Indeed, suppose that ${\bf x} \in (Y+ {\bf z}_1) \cap (Y+ {\bf z}_2) =  (Y + \psi({\bf d}_1)) \cap  (Y + \psi({\bf d}_2)),$ for any ${\bf d}_1 \neq {\bf d}_2 \in \mathcal{D}.$ Thus, 
\begin{align}\label{eq_schurzero}
{\bf x}= \psi({\bf c}_1) + \psi({\bf d}_1) = \psi({\bf c}_2) + \psi({\bf d}_2),
\end{align}
where $\psi({\bf c}_1), \psi({\bf c}_2) \in Y.$ By hypothesis, $({\bf c}_1 \ast {\bf d}_1)^{(1)} = {\bf 0} = ({\bf c}_2 \ast {\bf d}_2)^{(1)}$, and therefore, (\ref{eq_schurzero}) implies 
\begin{align}
\psi({\bf c}_1 \oplus {\bf d}_1) = \psi({\bf c}_2 \oplus {\bf d}_2),
\end{align}
and ${\bf c}_1 \oplus {\bf d}_1 = {\bf c}_2 \oplus {\bf d}_2,$ because $\psi$ is injective. However, each ${\bf d}_i,~ i=1,2$ represents a distinct coset of $\mathcal{C}$ and the equality in the above equation means that ${\bf d}_1={\bf d}_2,$ which is a contradiction. Therefore  $(Y+ {\bf z}_1) \cap (Y+ {\bf z}_2) = \varnothing$ and the proof is complete.
\end{proof}

	A particular case where the hypothesis of Theorem~\ref{thm_zeroschur} is satisfied is when each coset of $\cbar$ has a representative in the form $({\bf 0}, \dots, {\bf c}_L) \in \mathcal{D},$ for all ${\bf c}_L \in \mathcal{C}_L.$ Hence, the level shifts of all Schur products will be zero. An example follows next.

\begin{example} 
Consider the linear code $\mathcal{C}$ in Example~\ref{ex:simple}, together with $X$ and $Y$.
Set ${\bf d_1} = {\bf 0}$ and among the nonzero possible choices for the remaining coset representative in $\mathcal{D},$ we select ${\bf d}_2 = (0,0,0,0,1,0).$ Notice that the Schur product between ${\bf d}_2$ and any element of $\mathcal{C}$ is either zero or ${\bf d}_2,$ whose level shift is zero and the condition of Theorem~\ref{thm_zeroschur} is satisfied. $(Y,Z)$ is a tiling of $X,$ with
$Z=\{(0,0),(4,0)\}.$
\end{example}
	
    The result of Theorem~\ref{thm_zeroschur} is particularly useful when we are working with a small number of levels. Theorem~\ref{thm_necessary} gives a necessary condition for $(Y,Z)$ be a tiling of $X.$
	
\begin{theorem}\label{thm_necessary} If $(Y,Z)$ is a tiling of $X$ and at most the second Schur level shift is vanishing, then the Schur level shift $({\bf c} \ast {\bf d}_i)^{(1)} \in  \cbar,$ for all ${\bf c} \in \mathcal{C}$ and ${\bf d}_i \in \mathcal{D}.$
\end{theorem}

\begin{proof} Because $(Y,Z)$ is a tiling of $X,$ there exist ${\bf y} \in Y$ and ${\bf z}_i \in Z,$ such that if we take ${\bf x} \in \Gamma_C,$ we can write
\begin{align}
{\bf x} = & {\bf y} + {\bf z}_i =  \underbrace{\psi({\bf c})}_{\in Y} + \underbrace{\psi({\bf d}_i)}_{{\in Z}}, 
\end{align}
for some ${\bf c} \in \mathcal{C}, {\bf d}_i \in \mathcal{D}.$ Hence, 
\begin{align*}
{\bf x} = & \psi({\bf c} \oplus {\bf d}_i) + 2\psi({\bf c} \ast {\bf d}_i) = \psi({\bf c} \oplus {\bf d}_i) + \psi\big(({\bf c} \ast {\bf d}_i)^{(1)}\big),
\end{align*}
for ${\bf c \in \mathcal{C}}.$ If $\big({\bf c} \ast {\bf d}_i\big)^{(1)} = {\bf 0},$ then evidently $\big({\bf c} \ast {\bf d}_i\big)^{(1)} \in \cbar.$ Otherwise, if $({\bf c} \ast {\bf d}_i)^{(1)} \neq {\bf 0},$ then
\begin{align*}
{\bf x} = & \psi({\bf c} \oplus {\bf d}_i) + \psi\big(({\bf c} \ast {\bf d}_i)^{(1)}\big) \\
= & \psi({\bf c} \oplus {\bf d}_i \oplus ({\bf c} \ast {\bf d}_i)^{(1)}) + 2\psi({\bf c} \oplus {\bf d}_i \ast ({\bf c} \ast {\bf d}_i^{(1)}))  \\
= & \psi({\bf c} \oplus {\bf d}_i \oplus ({\bf c} \ast {\bf d}_i)^{(1)})) + \psi\underbrace{(({\bf c} \oplus {\bf d}_i \ast ({\bf c} \ast {\bf d}_i)^{(1)})^{(1)})}_{{\bf 0}~ \text{by hypothesis}},
\end{align*}
which implies that ${\bf c} \oplus {\bf d}_i \oplus ({\bf c} \ast {\bf d}_i)^{(1)} \in \cbar$ and consequently $({\bf c} \ast {\bf d}_i)^{(1)} \in \cbar.$ 
\end{proof}
	
\begin{example}\label{ex_lattice} Consider $X=\mathbb{Z}_{16}$ and $Y=\{0,5,10,15\}.$ Thus, it is true for this setting of points that
\begin{align}\label{eq_classes}
X = Y + \underbrace{\{0,4,8,12\}}_{Z}.
\end{align}
and $(Y,Z)$ is a tiling of $X.$ The linear codes generating such constellations are, respectively to $Y$ and $X,$
\begin{gather*}
\mathcal{C} = \{(0,0,0,0),(1,0,1,0),(0,1,0,1),(1,1,1,1)\} \subset \mathbb{F}_2^{4} \nonumber \\
\cbar = \mathbb{F}_2^{4}.
\end{gather*}
The classes in the set Z according to~\eqref{eq_classes} are written as $4 = \psi(0,0,1,0),$ $8 = \psi(0,0,0,1),$ and $12 = \psi(0,0,1,1).$ By easy calculations, we can check that at most the second Schur level shifts between the class representatives and codewords of $\mathcal{C}$ are zero. In this particular case, $X$ is an abelian group (and  $\Gamma_C = X + 16\mathbb{Z}$ is a lattice). Theorem~\ref{thm_necessary} then applies.
\end{example}
	
	We could keep calculating the recursion proposed in Theorem~\ref{thm_necessary} and exploiting conditional consequences of the tiling, but that result already implies that there might be significant contributions by considering $X$ to be an abelian group (as in Example~\ref{ex_lattice}), which is the topic of the next section.

\section{When the finer constellation is a lattice} \label{Sec:5}

	From now on, we assume that the finer constellation $\Gamma_C$ is a lattice. Then, since $\Gamma_C=X+2^L\mathbb{Z}^n,$ $X$ will be taken as a finite abelian group of  $\mathbb{Z}_{2^L}^n.$ The result below is intuitive.

\begin{proposition} $\Gamma_C=X+2^L\mathbb{Z}^n$ is a lattice if and only if $X$ is a finite abelian group of $\mathbb{Z}_{2^L}^n.$
\end{proposition}
	
	We are interested in finding conditions on the set $Y$ such that it tiles $X$ and consequently, we will be able to conclude that $\Gamma_{C^\star}$ tiles $\Gamma_C.$ In order to explore this algebraic direction, we need the next definition.
	
\begin{definition} If $A$ is a subset of $B$, we define the set
\begin{align*}
A-A=\{{\bf b} \in B: ~ \exists ~ {\bf a_1}, {\bf a_2} \in A ~ \text{such that} ~ {\bf b}={\bf a_1}-{\bf a_2}\}. 
\end{align*}
\end{definition}
	
	
	The following result originates from group theory \cite{dinitz06} (more details about factoring abelian groups by sets can be found in \cite{newman77} \cite{sands79} \cite{szabo09}).

\begin{theorem}\label{thm_tiling_group} \cite{dinitz06} Let $Y,Z \subseteq X.$ Then $(Y,Z)$ is a tiling of $X$ if and only if $(Y-Y) \cap (Z-Z) = \{{\bf 0}\}$ and $|X| = |Y||Z|.$
\end{theorem}
\begin{proof} $(\Rightarrow)$ Let $(Y,Z)$ be a tiling of $X.$  If $(Y-Y) \cap (Z-Z) \neq \{{\bf 0}\},$ then there exist distinct elements ${\bf y}_1, {\bf y}_2 \in Y, {\bf z}_1, {\bf z}_2 \in Z$ such that ${\bf y}_1-{\bf y}_2={\bf z}_1-{\bf z}_2 \neq {\bf 0}.$ But then ${\bf y}_1+{\bf z}_2={\bf y}_2+{\bf z}_1,$ which contradicts the fact that $(Y,Z)$ is a tiling of $X.$ On the other hand, if $(Y-Y) \cap (Z-Z) = \{{\bf 0}\}$ and $|X| \neq |Y||Z|$, then  $|X| > |Y||Z|$ and there exist some element of $X$ that is not in $Y + Z,$ which is a contradiction of the fact that $(Y,Z)$ is a tiling of $X.$

\noindent $(\Leftarrow)$ Now we assume that $(Y-Y) \cap (Z-Z) = \{{\bf 0}\}$ and ${\bf x} = {\bf y}_1+{\bf z}_1={\bf y}_2+{\bf z}_2$ for some ${\bf x} \in X$. Then ${\bf y}_1-{\bf y}_2={\bf z}_2-{\bf z}_1={\bf 0}$ and ${\bf y}_1={\bf y}_2,~ {\bf z}_1 = {\bf z}_2.$ This means that the representation of ${\bf x}$ is unique. The fact that $|X| = |Y||Z|$ allows us to guarantee that $(Y,Z)$ is a tiling of $X.$
\end{proof}

    When we consider $Y=\psi(\mathcal{C}),$ and $Z=\psi(\mathcal{D})$ as the set of class representatives, the condition $|X| = |Y||Z|$ is always true since the cardinalities of $X,Y$ and $Z$ are powers of $2,$ as stated previously. Therefore, the only condition that needs to be verified is $(Y-Y) \cap (Z-Z) = \{ {\bf 0}\}.$
	
	Consider an element $w \in (Y-Y) \cap (Z-Z).$ Then, for ${\bf c}_1, {\bf c}_2 \in \mathcal{C},$ ${\bf d}_1, {\bf d}_2 \in \mathcal{D}$,
\begin{align}\label{eq_intersection}
w = \psi({\bf c}_1) - \psi({\bf c}_2) = \psi({\bf d}_1) - \psi({\bf d}_2).
\end{align}
Elaborating \eqref{eq_intersection}, we get
{\small \begin{align*}
\psi({\bf c}_1) + \psi({\bf d}_2)  & = \psi({\bf c}_2) + \psi({\bf d}_1) \nonumber \\
\psi({\bf c}_1 \oplus {\bf d}_2) + 2\psi({\bf c}_1 \ast {\bf d}_2) & = \psi({\bf c}_2 \oplus {\bf d}_1) + 2\psi({\bf c}_2 \ast {\bf d}_1) \nonumber \\
\psi({\bf c}_1 \oplus {\bf d}_2) + \psi\big(\underbrace{({\bf c}_1 \ast {\bf d}_2)^{(1)}}_{\in \cbar}\big) & = \psi({\bf c}_2 \oplus {\bf d}_1) + \psi\big(\underbrace{({\bf c}_2 \ast {\bf d}_1)^{(1)}}_{\in \cbar}\big),
\end{align*}}
because $X = \psi(\cbar)$ is an additive group and the chain of codes $\mathcal{C}_1 \subseteq \dots \subseteq \mathcal{C}_L$ is closed under the Schur product. By proceeding recursively, until the Schur product is zero, we will end up with the following representation
\begin{align}\label{eq_schurs}
\psi({\bf c}_1 \oplus ({\bf d}_2 \oplus {\bf s})) = \psi({\bf c}_2 \oplus ({\bf d}_1 \oplus {\bf \tilde{s}})),
\end{align} 
where ${\bf s}, {\bf \tilde{s}}$ represent the sums in $\mathbb{F}_2^n$ of the recursive Schur level shifts of elements of $\mathcal{C}$ with elements of $\cbar \setminus \mathcal{C}$.
Since $\psi$ is injective, we can conclude from~\eqref{eq_schurs} that ${\bf c}_1 \oplus ({\bf d}_2 \oplus {\bf s}) = {\bf c}_2 \oplus ({\bf d}_1 \oplus {\bf \tilde{s}}).$ 

	If ${\bf d}_2 \oplus {\bf s}, {\bf d}_1 \oplus {\bf \tilde{s}} \in (\cbar \setminus \mathcal{C}) \cup \{{\bf 0}\},$ we can conclude that ${\bf c}_1 = {\bf c}_2$ and  $w={\bf 0}.$ Consequently, the result below has just been demonstrated.

\begin{theorem}\label{thm_tiling_code} Let ${\bf s}$ and ${\bf \tilde{s}}$ be defined as in \eqref{eq_schurs}. 
If ${\bf s},{\bf \tilde{s}} \in (\cbar \setminus \mathcal{C}) \cup \{{\bf 0}\},$ then $(Y,Z)$ is a tiling of $X.$
\end{theorem}

    The condition of Theorem~\ref{thm_tiling_code} might be hard to check for dense constellations. A practical question that remains unsolved is: how can we select the set $Z,$ given that $Y=\psi(\mathcal{C})$ is fixed, based on the linear codes? To understand the process we carry out as follows.
    
    Initially, consider the set $\mathcal{S} = (\cbar \setminus \mathcal{C}) \cup \{{\bf 0}\},$ where we must select the elements to compose $\mathcal{D}$ and $Z=\psi(\mathcal{D}).$ 

\begin{enumerate}
    \item For every element ${\bf c} \in \mathcal{C},$ find the correspondent ${\bf c'} \in \cbar$ such that $-\psi({\bf c})=\psi({\bf c'}),$ according to Eqs.~\eqref{eq_negative0} and \eqref{eq_negative}. Remove all these elements ${\bf c'}$ from the set $\mathcal{S}.$
    \item Consider now the elements ${\bf \overline{c}} \in \cbar$ such that $\psi({\bf c}) -\psi({\bf \tilde{c}}) = \psi({\bf \overline{c}}),$ for all nonzero ${\bf c} \neq {\bf \tilde{c}} \in \mathcal{C}.$ Remove all these elements ${\bf \overline{c}}$ from the set $\mathcal{S}$ as well. 
    \item Thus, $\mathcal{D} \subseteq \mathcal{S}.$ If $\mathcal{D} \subsetneq \mathcal{S},$ then one must select $2^{\overline{m}-m}-1$ nonzero elements in $\mathcal{S}$ to form $\mathcal{D},$ such that for $Z=\psi(\mathcal{D})$ the relation $(Y-Y) \cup (Z-Z) = \{{\bf 0}\}$ holds.
\end{enumerate}

    This process is shown with details in the next example.  

\begin{example} Consider $X=\mathbb{Z}_{16},$ $Y=\{0,3,12,15\}$ and $Z=\{0,\textcolor{red}{{\bf 2}},\textcolor{green}{{\bf 8}},\textcolor{blue}{{\bf 10}}\}$ (colors refer to $X=Y+Z$ in Figure~\ref{fig:ex_tiling}), 
 Following Theorem~\ref{thm_tiling_group}, we need to check that for $Y$ and $Z,$ $(Y-Y) \cap (Z-Z)=\{{\bf 0}\}$. 
\begin{align*}
Y-Y=& \{0,1,3,4,7,9,12,13,15\}, \nonumber \\
Z-Z= &\{0,2,6,8,10,14\}.
\end{align*}
Clearly $(Y-Y) \cap (Z-Z)=\{{\bf 0}\}$ and $(Y,Z)$ is a tiling of $X,$ illustrated in Figure~\ref{fig:ex_tiling}. 

\begin{figure}[h!]
    \centering
    \includegraphics[scale=0.199]{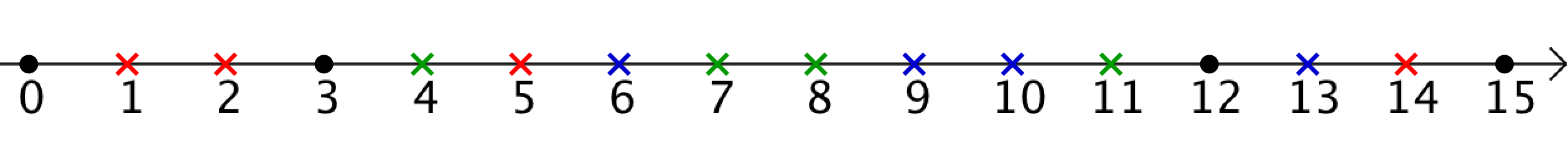}
    \caption{Tiling of $\mathbb{Z}_{16}$ by $(Y,Z).$}
    \label{fig:ex_tiling}
\end{figure}

    Given $X$ and $Y,$ we will explain how one can choose $Z$ properly. Recall that $\mathcal{C}=\{(0,0,0,0),(1,1,0,0),(0,0,1,1),$ $(1,1,1,1)\}$ and $\cbar = \mathbb{F}_2^4,$ here $n=1$ and $L=4.$ Start with $\mathcal{S}=(\cbar \setminus \mathcal{C}) \cup \{{\bf 0}\},$ i.e., 
\begin{align*}
\mathcal{S} =  & \{(0,0,0,0),(1,0,0,0),(0,1,0,0), (0,0,1,0), (0,0,0,1),\\
 & (1,0,1,1),(0,1,1,1),(1,1,0,1),(1,1,1,0),(0,1,0,1), \\
& (1,0,0,1),(1,0,1,0), (0,1,1,0)\}.
\end{align*}

    We calculate ${\bf c'}$ such that $-\psi({\bf c}) = \psi({\bf c'})$ and remove such elements from $\mathcal{S}.$ Then
\begin{align*}
\mathcal{S} =  & \{(0,0,0,0),(0,1,0,0),(0,0,0,1), (0,1,1,1),(1,1,0,1), \\
 & (1,1,1,0),(1,0,1,0),(0,1,0,1),(1,0,0,1), (0,1,1,0)\}.
\end{align*}

    Moving forward, we exclude from $\mathcal{S}$ the elements ${\bf \overline{c}},$ and
\begin{align*}
\mathcal{S} =  & \{(0,0,0,0),(0,1,0,0),(0,0,0,1),(0,1,1,1),(1,1,0,1), \\
 & (1,0,1,0),(0,1,0,1),(0,1,1,0)\}.
\end{align*}

    Observe that $2^3=|\mathcal{S}|>|\mathcal{D}|=2^2.$ Since $\mathcal{C} \oplus (0,0,0,1) \simeq \mathcal{C} \oplus (1,1,0,1),$ and $(0,0,0,1)$ is a straightforward choice for representing its class, as the Schur level shift with any element of $\mathcal{C}$ is always zero, we update $\mathcal{S} \setminus \{(1,1,0,1)\}.$ Thus, 
\begin{align*}
\mathcal{D} \subset \mathcal{S} = & \{(0,0,0,0),(0,1,0,0),(0,0,0,1),(0,1,1,1), \\
 & (1,0,1,0), (0,1,0,1),(0,1,1,0)\}
\end{align*}
which means that $Z \subset \{0,2,5,6,8,10,14\}.$ Because $8$ is fixed, $5$ can also be excluded from $Z,$ because $8-5 \in (Y-Y) \cup (Z-Z).$ The only two possible choices for $Z$ are $Z=\{0,\textcolor{red}{{\bf 2}},\textcolor{green}{{\bf 8}},\textcolor{blue}{{\bf 10}}\}$ (Figure~\ref{fig:ex_tiling}) or $\tilde{Z}=\{0,\textcolor{red}{\bf 6},\textcolor{green}{\bf 8},\textcolor{blue}{\bf 14}\}$ (Figure~\ref{fig:ex_tiling_}).

\begin{figure}[H]
    \centering
    \includegraphics[scale=0.195]{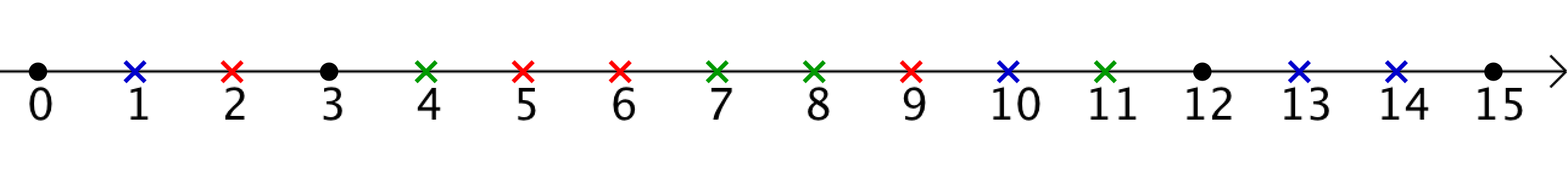}
    \caption{Tiling of $\mathbb{Z}_{16}$ by $(Y,\tilde{Z}).$}
    \label{fig:ex_tiling_}
\end{figure}

    It is interesting that we do not have equivalence of classes, but both $(Y,Z)$ and $(Y,\tilde{Z})$ are tilings of $X,$ which means that the tiling is not unique in general.
\end{example} 


\section{Conclusion and future work} \label{Sec:Conc}

    We have investigated the tiling of a subset $X \subseteq \mathbb{Z}_{2^L}^n$ that in general is not an abelian group. $X$ is generated by a linear code, thus we use the tiling of the underlying code to derive conditions on the tiling of $X$. We have also applied a result about finite abelian group factorization to explore the tiling of $X,$ for the lattice case. As a consequence of that, it is possible to tile $\Gamma_C$ with a subset $\Gamma_{C^\star}$ under certain assumptions.
	
	


\end{document}